%% file: main.tex
\documentclass[sigconf]{acmart}

\AtBeginDocument{%
  \providecommand\BibTeX{{%
    \normalfont B\kern-0.5em{\scshape i\kern-0.25em b}\kern-0.8em\TeX}}}

\copyrightyear{2023}
\acmYear{2023}
\setcopyright{acmlicensed}\acmConference[KDD '23]{Proceedings of the 29th ACM SIGKDD Conference on Knowledge Discovery and Data Mining}{August 6--10, 2023}{Long Beach, CA, USA}
\acmBooktitle{Proceedings of the 29th ACM SIGKDD Conference on Knowledge Discovery and Data Mining (KDD '23), August 6--10, 2023, Long Beach, CA, USA}
\acmPrice{15.00}
\acmDOI{10.1145/3580305.3599352}
\acmISBN{979-8-4007-0103-0/23/08}

\usepackage{common}
\usepackage{mymacro}
\usepackage{balance}
\usepackage[hide]{notes}
\allowdisplaybreaks

\newif\ifapx 

\newif\ifsupp 
\supptrue 
\ifsupp\else  \fi 

\newcommand\apxmode{append}
\ifapx\else
\renewcommand\apxmode{inline}
\fi
\ifsupp\else
\renewcommand\apxmode{strip}
\fi
\usepackage[appendix=\apxmode]{apxproof} 
\renewcommand*{\appendixprelim}{\clearpage} 
\renewcommand{\appendixsectionformat}[2]{Proofs for Section~#1}

\newtheorem{problem}{Problem}

\newtheorem{example}{Example}

\newtheorem{definition}{Definition}
\newtheorem{assumption}{Assumption}

\newtheorem{theorem}{Theorem}
\newtheorem{lemma}[theorem]{Lemma}

\crefname{algocf}{Algorithm}{Algorithms}

\begin{document}

\title[Finding Favourite Tuples on Data Streams with Provably Few Comparisons]{Finding Favourite Tuples on Data Streams with\\Provably Few Comparisons}

\author{Guangyi Zhang}
\authornote{This work was done while the author was with KTH Royal Institute of Technology.}
\affiliation{%
  \institution{Shenzhen Institute of Computing Sciences} 
  \city{Shenzhen}
  \country{China}
}
\email{zhangguangyi@sics.ac.cn}

\author{Nikolaj Tatti}
\affiliation{%
  \institution{HIIT, University of Helsinki}
  \city{Helsinki}
  \country{Finland}}
\email{nikolaj.tatti@helsinki.fi}

\author{Aristides Gionis}
\affiliation{%
  \institution{KTH Royal Institute of Technology}
  \city{Stockholm}
  \country{Sweden}
}
\email{argioni@kth.se}


\begin{abstract}
One of the most fundamental tasks in data science is to assist 
a user with unknown preferences in finding high-utility tuples within a large database.
To accurately elicit the unknown user preferences, 
a widely-adopted way is by asking the user to compare pairs of tuples.
In this paper, we study the problem of identifying one or more high-utility tuples 
by adaptively receiving user input on a minimum number of pairwise comparisons.
We devise a single-pass streaming algorithm,
which processes each tuple in the stream at most once, 
while ensuring that the memory size and the number of requested comparisons 
are in the worst case logarithmic in $n$,
where $n$ is the number of all tuples.
An important variant of the problem, 
which can help to reduce human error in comparisons,
is to allow users to declare ties when confronted with pairs of tuples of nearly equal utility.
We show that the theoretical guarantees of our method can be maintained
for this important problem variant. 
In addition, we show how to enhance existing pruning techniques in the literature 
by leveraging powerful tools from mathematical programming.
Finally, we systematically evaluate all proposed algorithms over both synthetic and 
real-life datasets,
examine their scalability, and 
demonstrate their superior performance over existing methods. 
\end{abstract}

\ccsdesc[500]{Information systems~Users and interactive retrieval}
\ccsdesc[500]{Theory of computation~Database theory}
\ccsdesc[300]{Theory of computation~Active learning}

\keywords{preference learning, user interaction, streaming algorithms}


\maketitle

\ifsupp 
\ifapx
	\begin{toappendix} 
	\input{supp}

	\end{toappendix}
\fi
\fi 

\section{Introduction}
\label{se:intro}

\input{intro}

\section{Problem definition}
\label{se:def}
\input{def}

\section{Related work}
\label{se:related}
\input{related}

\section{Finding a tuple: oracle with no ties}
\label{se:single}

\input{irm}

\section{Finding a tuple: oracle with ties}
\label{se:robust}
\input{robust}

\section{Improving baseline filters}
\label{se:extensions-short}
\input{extensions-short}

\section{Experimental evaluation}
\label{se:experiment}
\input{experiment}

\section{Conclusion}
\label{se:conclusion}
\input{conclusion}

\begin{acks}
This research is supported by the Academy of Finland projects MALSOME (343045) and MLDB (325117),
the ERC Advanced Grant REBOUND (834862),
the EC H2020 RIA project SoBigData++ (871042),
and the Wallenberg AI, Autonomous Systems and Software Program (WASP)
funded by the Knut and Alice Wallenberg Foundation.
\end{acks}

\note[SVM]{Just a remark that SVM is a special instance of QP.
Different from constrained least squares, 
it minimizes the $\ell_2$ norm the coefficient vector for a linear function 
subject to a set of linear inequalities (i.e., hard margin constraints), one for each data point.}


\newpage
\balance
\bibliographystyle{ACM-Reference-Format}
\bibliography{references}

\ifsupp 
\ifapx \else
	\clearpage
	\appendix
	\input{supp}
\fi
\fi 

\end{document}

%% file: supp.tex

\section{Proofs for Section \ref{se:single}}
\label{se:single:proofs}
\input{irm-proofs}

\section{Proofs for Section \ref{se:robust}}
\label{se:robust:proofs}
\input{robust-proofs}

\section{Improving baseline filters}
\label{se:extensions}
\input{extensions}

\section{Additional experiments}
\label{se:experiment:supp}
\input{experiment-supp}

%% file: irm-proofs.tex
\citet{kane2018generalized} proved a powerful local lemma, which states that
among a sufficiently large set of vectors from the unit $d$-ball $\ball^d$, 
there must exist some vector that can be approximately represented as a special non-negative linear combination of others.
\begin{lemma}[\cite{kane2018generalized}, Claim 15]
\label{lemma:comb}
Given $\vecx_1,\ldots,\vecx_\infd \in \ball^d$, for any $\Creg > 0$, if $\infd \ge \Cinfd$,
then there exists $a \in [\infd]$ such that
\begin{align}
	\vecx_a = \vecx_1 + \sum_{j = 1}^{a-2} \alpha_j (\vecx_{j+1} - \vecx_j) + \vece,
	\label{eq:comb}
\end{align}
where $\norm{\vece}_2 \le \Creg$ and $\alpha_j \in \{0,1,2\}$.
\end{lemma}

Let $S = \{\vecx_1, \ldots, \vecx_\infd\}$.
\cref{lemma:comb} can be easily extended to hold for the intrinsic dimension of $S$,
by first applying \cref{lemma:comb} to the minimal representation $\vecy_1,\ldots,\vecy_\infd \in \ball^{d'}$ of $S$.
\note{
Suppose
\[
	\vecy_a = \vecy_1 + \sum_{j = 1}^{a-2} \alpha_j (\vecy_{j+1} - \vecy_j) + \vece',
\]
where $\norm{\vece'}_2 \le \Creg$.
By definition, we know that 
$\vecx_i = M \vecy_i$ where $M \in \reals^{d,d'}$ and columns of $M$ are orthonomal.
Then, we have $\norm{\vece}_2 \le \Creg$,
where $\vece = M\vece'$.
}

In \cref{lemma:comb}, we have $S-\vecx_a \infer \vecx_a$,
where $S-\vecx$ is a shorthand for $S \setminus \{\vecx\}$.
Note that this is exactly the condition we use in Step~\ref{step:filter} in \cref{alg:single} for pruning.
Denote by $\filter(S)$ the set of all such pruned tuples, i.e.,
\begin{align}
	\filter(S) = \{ \vecx \in \reals^d: S \infer \vecx \}.
\end{align}

Given any set $S$ of size $4\infd$,
at least 3/4 fraction of $S$ can be pruned by other tuples in $S$, by repeatedly applying \cref{lemma:comb}.
\begin{lemma}
\label{lemma:4t}
Given a sorted set $S$ of size at least $4\infd$, where $\infd = \Cinfd$, we have
\[
	|\{ \vecx \in S: S-\vecx \infer \vecx \}| \ge \frac34 |S|.
\]
\end{lemma}
\begin{proof}
since $|S| \ge 4\infd$,
we can apply \cref{lemma:comb} repeatedly to $S$ until only \infd entries remain.
\end{proof}
\note{Importantly, $\{ \vecx \in S: S-\vecx \infer \vecx \}$ do not depend on the arrival order of tuples in $S$.}

As a consequence of \cref{lemma:4t},
the same fraction of current tuples \X can be pruned by a random sample set $S$ of a sufficient size in expectation.
\begin{lemma}
\label{lemma:filter}
Given a set of tuples \X, and
a random sample set $S \subseteq \X$ of size $4\infd$ where $\infd = \ceil{\Cinfd}$, 
we have
\[
	\expect \left[ |\filter(S) \cap \X| \right] \ge \frac34 |\X|,
\]
where the expectation is taken over $S$.
\end{lemma}
\ifapx
\begin{proof}\deferredproof\end{proof}
\begin{appendixproof}[Proof of \cref{lemma:filter}]
\else
\begin{appendixproof}
\fi
The proof is by a symmetrization argument introduced by \citet{kane2017active}.
Let $\vecy$ be the last tuple added into $S$. Write $T = S - \vecy$.
Given $T$, the distribution of $\vecy$ is a uniform distribution from $\X \setminus T$.
Let $\vecx$ be a random sample from $X$. 
Since $T \subseteq \filter(T)$, we have
\[
	\Pr \left[  T \infer \vecx \mid T \right] \geq
	\Pr \left[  T \infer \vecy \mid T \right].
\]
Then,
\begin{align*}
    \expect_S \left[ \frac{|\filter(S) \cap \X|}{|\X|} \right] 
    &= \expect_{S} \left[ \Pr \left[ S \infer \vecx \mid S \right] \right] \\
    &= \expect_{S} \left[ \Pr \left[ T + \vecy \infer \vecx \mid S \right] \right] \\
    &\ge \expect_{S} \left[ \Pr \left[ T \infer \vecx \mid S \right] \right] \\
    &= \expect_{T} \left[ \Pr \left[ T \infer \vecx \mid T \right] \right] \\
    &\ge \expect_{T} \left[ \Pr \left[ T \infer \vecy \mid T \right] \right] \\
    &= \expect_{S} \left[ \indicator \left[ S - \vecy \infer \vecy \mid S \right] \right].
\end{align*}

In order to bound the right-hand side, notice that 
when conditioned on $S$, every permutation of $S$ is equally probable over a random-order stream, 
which implies that every tuple in $S$ is equally probable to be the last tuple \vecy.
Hence, we have $\Pr \left[ S - \vecy \infer \vecy \mid S \right] \geq 3/4$ by \cref{lemma:4t},
proving immediately the claim.
\end{appendixproof}
\note{It is important to note that, the above proof doesn't imply that
for \emph{every} $S$, when conditioning on it, one can achieve 
\[
	\Pr \left[ S \infer \vecx \mid S \right] \ge 3/4,
\]
that is, every $S$ can prune at least 3/4 of \X, which is too strong to be true.
After we replace \vecx by \vecy, the 3/4 is obtained by taking average over \vecy, which is part of outer expectation over $S$.
}

Another important issue to handle is to ensure that our pruning strategy will not discard all feasible tuples.
This is prevented by keeping track of the best tuple in any sample set so far, and guaranteed by \cref{thm:anytime}.
\begin{proof}[Proof of \cref{thm:anytime}]
Denote by \D all tuples that have arrived so far.
Suppose $\vecx^*$ is the best tuple among \D.
Tuple $\vecx^*$ is either collected into our sample sets, 
or pruned by some sample set $S$.
In the former case, our statement is trivially true.
In the latter case, suppose $S = \{\vecx_1,\ldots\}$, where $\vecx_1$ is the best tuple in $S$.
If $\vecx_1$ is feasible, then $\bestsofar$ is feasible as well, as it is at least as good as $\vecx_1$.
If $\vecx_1$ is infeasible, i.e., $\util(\vecx^*) - \util(\vecx_1) > \Creg$, then $\vecx^*$ cannot be pruned by $S$ by design, a contradiction.
This completes the proof.
\end{proof}

Before proving \cref{thm:num_of_S},
we briefly summarize the hypergeometric tail inequality below~\citep{skala2013hypergeometric}.
\begin{lemma}[Hypergeometric tail inequality~\citep{skala2013hypergeometric}]
\label{lemma:tail}
Draw $n$ random balls without replacement from a universe of $N$ red and blue balls, and
let $i$ be a random variable of the number of red balls that are drawn.
Then, for any $t > 0$, we have 
\[
	\Pr[i \ge \expect[i] + t n] \le e^{-2t^2n},
\]
and
\[
	\Pr[i \le \expect[i] - t n] \le e^{-2t^2n}.
\]
\end{lemma}

\begin{proof}[Proof of \cref{thm:num_of_S}]
The feasibility of the returned tuple \bestsofar is due to \cref{thm:anytime}.
In the rest of the proof, we upper bound the size of every sample and the number of samples we keep in the sequence~\seq.

For any sample $S$ with at least $4\infd$ samples and any subset $\X \subseteq \D$, 
let $\X' = \filter(S) \cap \X$ and by \cref{lemma:filter} we have
$\expect [ |\X'| ] \ge \frac34 |\X|$.
In particular, let $\X = \Pl$ and we have 
$\expect[|\Pl'|] \ge \frac34 |\Pl|$ and $|\Pl| = \nPl$.
Then, 
\begin{align*}
	\Pr\left[ |\Pl'| < \frac58 |\Pl| \right] 
	&= \Pr\left[ |\Pl'| < \frac34 |\Pl| - \frac18 \nPl\right]  \\
	&\le \Pr\left[ |\Pl'| < \expect[|\Pl'| - \frac18 \nPl] \right] 
	\le e^{-2 \nPl /8^2},
\end{align*}
where the last step invokes \cref{lemma:tail}.
Since there can be at most \nD samples,
the probability that any sample fails to pass the pool test is upper bounded by $\nD e^{-2 \nPl /8^2}$.

We continue to upper bound the number of sample sets.
At most $\ceil{\log(\nD)}$ sample sets suffice if every sample can prune at least half of the remaining tuples.
Fix an arbitrary sample $S$, and let \X to be the set of remaining tuples.
The pool \Pl is a random sample from \X of size \nPl. Thus,
$\expect[|\Pl'|] / \nPl = |\X'| / |\X|$. Consequently, if  $|\X'| < |\X|/2$, then $\expect[|\Pl'|] < \nPl/2$
and
\begin{align*}
	\Pr\left[ |\Pl'| \ge \frac58 |\Pl| \right] 
	\le \Pr\left[ |\Pl'| \ge \expect[|\Pl'|] + \frac18 \nPl \right] 
	\le e^{-2 \nPl /8^2}.
\end{align*}
Similar to the above, 
the probability that any bad sample passes the test is upper bounded by $\nD e^{-2 \nPl /8^2}$.

Combining the two cases above,
the total failure probability is $2\nD e^{-2 \nPl /8^2} \leq 1/\nD$
Hence, with probability at least $1-1/\nD$,
it is sufficient to use $\ceil{\log(\nD)}$ sample sets, each with a size $4\infd$.
Keeping one sample set requires $4\infd \ceil{\log(4\infd)}$ comparisons.
Finally, finding the best tuple among all filters and the pool requires additional $\nPl + \ceil{\log(\nD)}$ comparisons.
\end{proof}

%% file: robust-proofs.tex
The proof is similar to that of \cref{thm:num_of_S},
except that we need a new proof for the key \cref{lemma:filter},
since in the presence of ties, we may not be able to totally sort a sample $S$.
Instead, we show that a partially sorted set $S$ of a sufficient size can also be effective in pruning.

From now on, we treat the sample $S$ as a \emph{sequence} instead of a set, 
as a different arrival order of $S$ may result in a different filter by \cref{alg:robust}.

\begin{lemma}
\label{lemma:4t-robust}
Let $S \subseteq \X$ be a sequence of length $16\infd \Ccmpmax$.
Let $\Gs$ be the groups constructed by \cref{alg:robust}.
Under \cref{asm:dirm}, we have
\[
	|\{ \vecx \in S: \Gs - \vecx \infersim \vecx \}| \ge \frac34 |S|,
\]
where $\infd = \Cinfd$, 
\Ccmpmax is the largest size of a pairwise \vecu-similar subset of \X, 
and $\Gs - \vecx$ are the groups with $\vecx$ removed from its group.
\end{lemma}
\ifapx
\begin{proof}\deferredproof\end{proof}
\begin{appendixproof}[Proof of \cref{lemma:4t-robust}]
\else
\begin{appendixproof}
\fi
Note that by the definition of \Ccmpmax, 
for any particular tuple $\vecx \in S$,
there are at most $2(\Ccmpmax-1)$ tuples that are \vecu-similar with tuple \vecx.
Thus, $\Gs$ must contain at least $8\infd$ groups, and
we split all groups in $\Gs$ into two parts, those with an odd index and those with an even index.

In each part, we can extract a totally sorted list $L$ of size at least \infd, by picking exactly one tuple from each group.
We remove one tuple $\vecw \in L$ from $S$ such that $L - \vecw \infer \vecw$, whose existence is guaranteed by \cref{lemma:comb}.
\cref{eq:filter-robust} guarantees that $\Gs - \vecx \infersim \vecx$.

We repeatedly do so until less than \infd groups remain in each part, 
which means that the number of remaining tuples is at most $2 \infd \Ccmpmax$ in each part.
As a result, we are able to remove at least $16\infd \Ccmpmax - 4 \infd \Ccmpmax$ tuples, 
concluding the claim.
\end{appendixproof}

Although the above lemma appears similar to \cref{lemma:4t}, 
a crucial difference is that the set of prunable tuples in $S$ now depends on the arrival order of~$S$,
which causes non-trivial technical challenges in the analysis.
A critical observation that enables our analysis is the following result.
\begin{lemma}
\label{lemma:bounded-bad-representatives}
Fix a sequence $S$ of size $16\infd \Ccmpmax$, there exist at least $\frac14 |S|$ tuples \vecz in $S$ that satisfy
\[
	S - \vecz \infersim \vecz.
\]
\end{lemma}
\ifapx
\begin{proof}\deferredproof\end{proof}
\begin{appendixproof}[Proof of \cref{lemma:bounded-bad-representatives}]
\else
\begin{appendixproof}
\fi
Let $\Gs$ be the groups constructed by \cref{alg:robust}.
Write
\[
	S' = \{ \vecx \in S: \Gs - \vecx \infersim \vecx \}.
\]
By \cref{lemma:4t-robust}, we know that $|S'| \ge \frac34 |S|$.
For an arbitrary tuple $\vecz \in S'$, suppose \vecz is assigned to a group $G \in \Gs$.
We call a tuple \vecz \emph{good} if $|G| = 1$ or \vecz is not a representative in $R$ in \cref{alg:robust}.
Let $\Gs'$ be the groups constructed by \cref{alg:robust} using $S - \vecz$. If $\vecz$ is good,
then $\Gs' = \Gs - \vecz$.
Therefore, for a good tuple \vecz we always have 
\[
	S - \vecz \infersim \vecz.
\]
By definition, it is easy to see that there are at most $|S|/2$ tuples in $S$ that are not good,
proving the lemma.
\end{appendixproof}

Denote by $\filtersim(S)$ the set of tuples that can be pruned by $S$, that is,
\begin{align}
	\filtersim(S) = \{ \vecx \in \reals^d: S \infersim \vecx \}.
\end{align}
We now prove a similar lemma to \cref{lemma:filter} by a generalized symmetrization argument over sequences.
\begin{lemma}
\label{lemma:filter-robust}
Given a set of tuples \X, and
a random sequence $S$
of at least $16\infd \Ccmpmax$
tuples from $\X$,
we have
\[
	\expect \left[ |\filtersim(S) \cap \X| \right] 
	\ge \frac14 |\X|,
\]
where $\infd = \Cinfd$, 
and \Ccmpmax is the largest size of a pairwise \vecu-similar subset of \X. 
Moreover, the expectation is taken over $S$.
\end{lemma}
\ifapx
\begin{proof}\deferredproof\end{proof}
\begin{appendixproof}[Proof of \cref{lemma:filter-robust}]
\else
\begin{appendixproof}
\fi
Let $\vecy$ be the last tuple added into $S$. Write $T = S - \vecy$.
Given $T$, the distribution of $\vecy$ is a uniform distribution from $\X \setminus T$.
Let $\vecx$ be a random sample from $X$. 
Since $T \subseteq \filtersim(T)$, we have
\[
	\Pr [  T \infersim \vecx \mid T ] \geq
	\Pr [  T \infersim \vecy \mid T ].
\]
Then,
\begin{align*}
    \expect_S \left[ \frac{|\filtersim(S) \cap \X|}{|\X|} \right] 
    &= \expect_{S} [ \Pr [ S \infersim \vecx \mid S ] ] \\
    &= \expect_{S} [ \Pr [ T + \vecy \infersim \vecx \mid S ] ] \\
    &\ge \expect_{S} [ \Pr [ T \infersim \vecx \mid S ] ] \\
    &= \expect_{T} [ \Pr [ T \infersim \vecx \mid T ] ] \\
    &\ge \expect_{T} [ \Pr [ T \infersim \vecy \mid T ] ] \\
    &= \expect_{S} [ \indicator [ S - \vecy \infersim \vecy \mid S ] ].
\end{align*}

Fix $S$, let $\vecz \in S$ be a uniformly random tuple in $S$, and we have
\begin{align*}
	\expect_{S} [ \indicator [ S - \vecy \infersim \vecy \mid S ] ]
	&= \expect_{S} [ \Pr [ S - \vecz \infersim \vecz \mid S ] ] \\
	&\ge 1/4,
\end{align*}
where the last step is by \cref{lemma:bounded-bad-representatives}, and the first step is due to double counting,
as every sequence $S$ appears $|S|$ times in the right-hand side,
completing the proof.
\end{appendixproof}

\ifapx
\begin{proof}[Proof of \cref{thm:num_of_S-robust}]
The proof is similar to \cref{thm:num_of_S} on a high level, and is deferred to Appendix.
\end{proof}
\begin{appendixproof}[Proof of \cref{thm:num_of_S-robust}]
\else
\begin{appendixproof}
\fi
The proof is similar to \cref{thm:num_of_S} on a high level.
We only elaborate on their differences.

We first prove the guarantee on the regret.
If the optimal tuple $\vecx^*$ is in the pool once the algorithm is done, then the regret is at most $\Ccmp$.
If $\vecx^*$ is not in the pool, then
the proof of \cref{thm:anytime} shows that there is $\vecx$ in one of the sample, say $S$,
that yields a regret of $\Creg/c$.
The top representative of that sample yields $\Creg/c + \Ccmp$ regret. Finally, the final top tuple
yields $\Creg/c + 2\Ccmp$ regret.

Next, we upper bound the size of every sample and the number of samples similarly to
the proof of \cref{thm:num_of_S}.
We require every sample to prune at least $1/8$ fraction of the remaining tuples instead of $1/2$,
which leads to a demand for $\ceil{\log_{8/7} (\nD)}$ samples.
The total failure probability is bounded by $2\nD e^{-2 \nPl /16^2} \leq 1/\nD$.
Consequently, with probability at least $1-1/\nD$,
we will use at most
$\ceil{\log_{8/7} (\nD)}$ sample sets, each with a size $16\infd \Ccmpmax$, at most.
\note[more details]{
For any sample $S$ with at least $16\infd \Ccmpmax$ samples and any subset $\X \subseteq \D$, 
let $\X' = \filtersim(S) \cap \X$ and by \cref{lemma:filter-robust} we have
$\expect [ |\X'| ] \ge \frac14 |\X|$.
In particular, let \Pl is a random subset of \X and we have 
$\expect[|\Pl'|] \ge \frac14 |\Pl|$.
Then, 
\begin{align*}
	\Pr\left[ |\Pl'| < \frac3{16} |\Pl| \right] 
	&= \Pr\left[ |\Pl'| < \frac14 |\Pl| - \frac1{16} |\Pl|\right]  \\
	&\le \Pr\left[ |\Pl'| < \expect[|\Pl'| - \frac1{16} \nPl] \right] \\
	&\le e^{-2 \nPl /16^2},
\end{align*}
where the last step invokes \cref{lemma:tail}.
Since there can be at most \nD samples,
the probability that any sample fails to pass the pool test is upper bounded by $\nD e^{-2 \nPl /16^2}$.

We continue to upper bound the number of sample sets.
At most $\ceil{\log_{8/7}(\nD)}$ sample sets suffice if every sample can prune at least 1/8 fraction of the remaining tuples.
Fix an arbitrary sample $S$, and let \X to be the set of remaining tuples.
The pool \Pl is a random sample from \X of size \nPl. Thus,
$\expect[|\Pl'|] / \nPl = |\X'| / |\X|$. Consequently, if  $|\X'| < |\X|/8$, then $\expect[|\Pl'|] < \nPl/8$
and
\begin{align*}
	\Pr\left[ |\Pl'| \ge \frac3{16} |\Pl| \right] 
	\le \Pr\left[ |\Pl'| \ge \expect[|\Pl'|] + \frac1{16} \nPl \right] 
	\le e^{-2 \nPl /16^2}.
\end{align*}
Similar to the above, 
the probability that any bad sample passes the test is upper bounded by $\nD e^{-2 \nPl /16^2}$.
}

Building one filter requires at most $\bigO(16\infd \Ccmpmax \log(16\infd \Ccmpmax))$ comparisons, because
sorting an new tuple \vecx within $R$ by binary search costs at most $\bigO(\log(16\infd \Ccmpmax))$ comparisons.
Finally, finding the best tuple among all filters and the pool requires additional $\nPl + \ceil{\log_{8/7} (\nD)}$ comparisons.
\end{appendixproof}

%% file: extensions.tex
In this section, 
we improve existing filters by \citet{xie2019strongly},
by using linear and quadratic programs.
Previously, their filters rely on explicit computation of convex hulls,
which is feasible only in very low dimension.
For example, the convex hull size, and consequently
the running time of these existing techniques, 
have an exponential dependence on $d$~\citep{barber1996quickhull}. 

\subsection{Improving \emph{constrained utility space} filter}
\label{sse:utility}

One of the most natural strategies is to iteratively compare a pair of random tuples.
The feasible space for the utility vector \vecu is constrained by the list of pairs 
$A = \{a_i\}$ that have been compared,
where $a_i = (\vecy,\vecz)$ such that $\util(\vecy) < \util(\vecz)$.
Note that every pair of tuples $\vecy,\vecz \in \D$ forms a halfspace in $\reals^d$, i.e.,
$\h = \{ \vecu \in \reals^d: \vecu^T(\vecy-\vecz) < 0 \}$.
Specifically, the unknown $\vecu \in \sphere$ is contained in the intersection $U$ of a set of halfspaces, one by each pair.

\citet[Lemma 5.3]{xie2019strongly} propose to prune a tuple \vecx if 
for \emph{every} possible $\vecu \in U$ there exists a tuple \vecw in some pair of $A$ such that $\util(\vecw) \ge \util(\vecx)$.
They first compute all extreme points of $U$, and then check if the condition holds for every extreme point.
However, this approach is highly inefficient, 
as potentially there is an exponential number of extreme points.

Instead, we propose to test the pruning condition by asking to find 
a vector $\vecu$ that satisfies
\begin{equation}
	\eqlp
\end{equation}

If there is no such vector $\vecu$ we prune $\vecx$.
This test can be done with a linear program (LP).
Note that the test is stronger than that by \citet{xie2019strongly} as it has been extended to handle $\epsilon$-regret.

We claim that a given tuple \vecx can be safely pruned if there is no vector~$\vecu$ satisfying LP~(\cref{eq:lp}).

\proplp*

\ifapx
\begin{proof}\deferredproof\end{proof}
\begin{appendixproof}[Proof of \cref{prop:lp}]
\else
\begin{appendixproof}
\fi
Let $\vecu$ be the utility vector.
The assumptions imply 
\begin{equation}
\label{eq:cond1}
	\vecu^T\vecx - \vecu^T \vecz > 0
	\quad\text{and}\quad
	\vecu^T((1-\Creg)\vecx - \vecz) >  0.
\end{equation}

Next, note that, by definition, for every $(\vecy,\vecz) \in A$,
\begin{equation}
\label{eq:cond2}
	\vecu^T(\vecz-\vecy) > 0.
\end{equation}

The inequalities in Eqs.~(\ref{eq:cond1})--(\ref{eq:cond2}) are all proper.
Consequently, we can scale $\vecu$ so that the left-hand sides in Eqs.~(\ref{eq:cond1})--(\ref{eq:cond2}) are at least 1,
that is, there exists a solution to LP~(\cref{eq:lp}).
\end{appendixproof}

Notice that the second set of constraints in LP~(\cref{eq:lp}) (i.e., $\vecu^T(\vecx-\vecz) \ge 1$) is redundant provided $\util(\vecx) \geq 0$.
Actually, even if $\util(\vecx) < 0$, the test only lets in \vecx that is slightly worse than the best tuple in $A$, 
which is unlikely since $\util(\vecx) < 0$.
Thus, in practice we recommend to omit the second set of constraints to speed up the test.

\note[sorted list vs. pairs]{
It is not true that 
sorting a list of tuples is a much more efficient way to generate compared pairs.
At least for the purpose of maintaining utility space.
There are quadratic number of pairs in a list, but the ``effective'' pairs are those made by actual comparisons.
}

A filter for maintaining the constrained utility space is conceptually different from the filter proposed in \cref{se:single}.
A small utility space of \vecu is the key for such a filter to be effective,
while a filter in \cref{se:single} maintains no explicit knowledge about \vecu and mainly relies on the geometry of the tuples.

\subsection{Improving \emph{conical hull of pairs} filter}
\label{sse:cone-by-pairs}

Another pruning strategy proposed by \citet[Lemma 5.6]{xie2019strongly}
is the following. 
Consider again a list of compared pairs $A = \{a_i\}$,
where $a_i = (\vecy,\vecz)$ such that $\util(\vecy) < \util(\vecz)$,
and consider a cone formed by all pairs in $A$.
A tuple \vecx can now be pruned
if there is another tuple \vecw kept by the algorithm,
such that
\[
	\vecx = \vecw + \sum_{a_i = (\vecy,\vecz) \in A} \beta_i (\vecy-\vecz)
	\quad\text{such that}\quad
	\beta_i \ge 0
	\quad\text{for all } i.
\]

Instead of actually constructing all facets of the conical hull, 
as done by \citet{xie2019strongly},
we propose to solve the following quadratic program (QP),
\begin{align*}
	\eqqppairs.
\end{align*}

If the optimal value of the QP is at most $\Creg$, we prune \vecx.

\propqppairs*
\ifapx
\begin{proof}\deferredproof\end{proof}
\begin{appendixproof}[Proof of \cref{prop:qp-pairs}]
\else
\begin{appendixproof}
\fi
We only discuss the case $\Creg = 0$.
When $\Creg > 0$, for any pruned tuple, there exists a tuple in some pair of $A$ that is at most a distance of \Creg away from it, and
thus $A$ maintains at least one $\Creg / c$-regret tuple.

The first sum in QP (\cref{eq:qp-pairs}) can be seen as an aggregated tuple by convex combination,
whose utility is no better than the top tuple in~$A$.
The second term only further decreases the utility of the first term.
Thus, if a tuple \vecx can be written as a sum of the first and second terms, 
its utility is no better than the top tuple in~$A$, and 
can be pruned.
\end{appendixproof}

Similar to \cref{eq:filter:lp}, a weaker but computationally more efficient filter can be used, by replacing the QP with an LP solver.
That is, we prune tuple \vecx if there is a solution to
\begin{align}
	\label{eq:lp-pairs}
	 \vecx = &\sum_{a_i = (\vecy,\vecz) \in A} \nu_{i} \, \vecz + \sum_{a_i = (\vecy,\vecz) \in A} \beta_i (\vecy-\vecz) \\
	 \text{such that}
	&\sum_{a_i = (\vecy,\vecz) \in A} \nu_{i} = 1 , 
	 \quad\text{and}\quad
	\nu_{i1}, \nu_{i2}, \beta_i \ge 0 \quad\text{for all } i. \nonumber
\end{align}

As a final remark about the above QP,
we compare its pruning power with that of the proposed filter (\cref{eq:filter}) in \cref{se:single}.
Obviously, its pruning power increases as the number of compared pairs in $A$ increases.
For a fixed integer $s$, 
a number of $s$ comparisons result in $s$ pairs for the above QP,
while in \cref{se:single},
$s$ comparisons can produce a sorted list of $s/\log(s)$ tuples and ${s/\log(s) \choose 2}$ pairs.
Hence, the above QP is less ``comparison-efficient'' than the one in \cref{se:single}.
Also, for a fixed number of compared pairs, the number of parameters is larger in QP~(\cref{eq:qp-pairs}) than in the proposed filter, 
which means that QP is more inefficient to solve.
These drawbacks are verified in our empirical study in the next section.

%% file: experiment-supp.tex
\smallskip
\para{Datasets.}
A summary of the real-life datasets we use for our evaluation can be found in \cref{tbl:datasets}.
The datasets contain a number of tuples up to 1M and a dimension up to 100.
Previous studies are mostly restricted to a smaller data size and a dimension size less than 10, and
a skyline operator is used to further reduce the data size in advance~\citep{qian2015learning,xie2019strongly,wang2021interactive}.
Note that running a skyline operator itself is already a time-consuming operation, especially for high-dimension data~\citep{borzsony2001skyline}, and 
becomes even more difficult to apply with limited memory size in the streaming setting.
Besides, a fundamental assumption made by a skyline operator, namely, 
pre-defined preference of all attributes, 
does not hold in our setting. 
According to this assumption, it is required to know beforehand whether an attribute is better with a larger or smaller value. 
This corresponds to knowing beforehand whether utility entry $\vecu_i$ is positive or negative for the i-th attribute. 
As we mentioned in \cref{se:def}, we do not make such an assumption about \vecu, and allow an arbitrary direction. 
This is reasonable, as preference towards some attributes may be diverse among different people. 
One example is the floor level in the housing market, where some may prefer a lower level, while others prefer higher. 
Hence, we do not pre-process the data with a skyline operator.

Details on the data generation process and the actual synthesized data 
can be found in our public Github repository.

\smallskip
\para{Baselines.}
We do not consider methods that synthesize fake tuples in pairwise comparisons, such as \citet{nanongkai2012interactive}.
Over a random-order stream, the algorithm by \citet{jamieson2011active} is the same as the baseline \hplp when adapted to find the top tuple instead of a full ranking.
The UH-Simplex method~\citep{xie2019strongly} that simulates the simplex method by pairwise comparisons is not included,
as it is mainly of theoretical interest, designed for offline computation, and has been shown to have inferior empirical performance compared to other baselines.
We do not consider baselines that iteratively compare a greedy pair (among all ${\nD \choose 2}$ pairs) with respect to some measure of interest,
such as \citet{qian2015learning,wang2021interactive},
because they are designed for offline computation and it is computationally prohibited to decide even the first greedy pair for the adopted datasets.

\smallskip
\para{Misc.}
We adopt the OSQP solver \citep{osqp} and the HIGHS LP solver~\citep{huangfu2018parallelizing}.
The maximum number of iterations for the solvers is set to 4000,
which is the default value in the OSQP solver.
All experiments were carried out on a server equipped with 24 processors of AMD Opteron(tm)
Processor 6172 (2.1 GHz), 62GB RAM, running Linux~2.6.\-32-754.35.1.el6.x86\_64.
The methods are implemented in Python~3.8.5.

\begin{figure}[!ht]
	\captionsetup[subfigure]{justification=justified,singlelinecheck=false,skip=-17pt}
    \centering
    \subcaptionbox{\label{fig:d-sz}}{
    	\pgfinput{prune-d-sz}
    }
    \subcaptionbox{\label{fig:tol-sz}}{
    	\pgfinput[trim={0 0 0 0.8cm},clip]{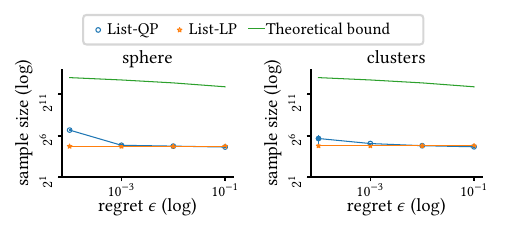}
    }
    \vspace{-0.7cm}
    \caption{Sample size required to prune half of tuples, 
    as a function of the data dimension (a), and 
    as a function of the regret parameter (b)}
    \label{fig:samplesize}
\end{figure}

\begin{figure}[!ht]
	\captionsetup[subfigure]{justification=justified,singlelinecheck=false,skip=-17pt}
	\centering
	\subcaptionbox{\label{fig:param:frac}}{
		\pgfinput[trim={0 0 0 0.1cm},clip]{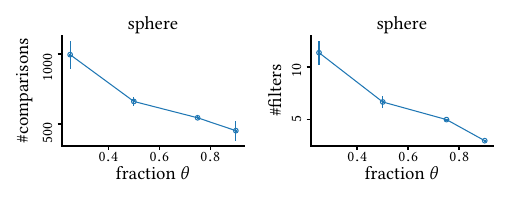}
	}\par\vspace{-0.5cm}
	\subcaptionbox{\label{fig:param:nPl}}{
		\pgfinput{param-nq-sth}
	}
	\vspace{-0.8cm}
	\caption{Effect of parameters on algorithm \listqp}
	\label{fig:param}
\end{figure}

\subsection{Effect of parameters}
\label{sse:ex:param}

Recall that in \cref{alg:framework}, a pool \Pl of \nPl tuples is used to test the performance of a new filter.
A new filter will be ready when it can prune at least a \CPl fraction of tuples in \Pl.
In \cref{fig:param}, we run \cref{alg:framework} with a \listqp filter on a dataset of 10k tuples.
We fix one parameter ($\nPl=100$ or $\CPl=0.5$) and vary the other.

Parameter \CPl roughly specifies the expected fraction of tuples a filter should be able to prune.
A larger \CPl implies a need for fewer filters but a larger sample size for each filter.
It is beneficial to use a large \CPl which leads a smaller number of comparisons overall.
Nevertheless, as we will see shortly, such a large filter can be time-consuming to run, especially when the dimension $d$ is large.

A larger value of \nPl improves the reliability of the testbed \Pl,
which helps reducing the number of comparisons.
However, a larger \nPl also results in longer time to run filters over the testbed \Pl.

%% file: intro.tex
\input{tbl-cmp}

One of the most fundamental tasks in data science is to assist 
a user with unknown preferences in finding high-utility tuples within a large database.
Such a task can be used, for example, 
for finding relevant papers in scientific literature, 
or recommending favorite movies to a user.
However, utility of tuples is highly personalized.
``\emph{One person's trash is another person's treasure,}'' as the saying goes.
Thus, a prerequisite to accomplishing this task is to efficiently and accurately 
elicit user preferences.

It has long been known, 
both from studies in psychology \citep{stewart2005absolute} 
as well as from personal experience, 
that humans are better at performing \emph{relative comparisons} 
than \emph{absolute assessments}.
For instance, 
it is typically easy for a user to select a favorite movie between two given movies,
while it is difficult to score the exact utility of a given movie.
This fact has been used in many applications,
such as classification \citep{haghiri2018comparison}, 
ranking \citep{liu2009learning},
and clustering~\citep{chatziafratis2021maximizing}.

In this paper we leverage the observation that humans are better at comparing
rather than scoring information items, 
and use relative comparisons to facilitate preference learning and help users find relevant tuples
in an interactive fashion, 
i.e., by adaptively asking users to compare pairs of tuples.
To cope with the issue of information overload, 
it is usually not necessary to identify all relevant tuples for a user.
Instead, if there exists a small set of high-utility tuples in the database,
a sensible goal is to identify at least one high-utility tuple 
by making a minimum number of comparisons.
In particular, assuming that a user acts as an \emph{oracle},
the number of requested comparisons, 
which measures the efficiency of preference learning, 
is known as \emph{query complexity}.

More specifically, in this paper we focus on the following setting. 
We consider a database \D consisting of \nD tuples, 
each represented as a point in $\reals^d$.
User preference is modeled by an unknown linear function 
on the numerical attributes of the data tuples.
Namely, we assume that a user is associated with an unknown utility 
vector~$\vecu \in \reals^d$, and
the utility of a tuple $\vecx \in \reals^d$ for that user is defined to~be 
\[
	\util(\vecx) = \vecu^T\vecx.
\]

A tuple \vecx is considered to be of high-utility 
if its utility is close to that of the best tuple,
or more precisely,
if compared to the best tuple its utility loss is bounded by 
an \Creg fraction of the best utility,
\[
	\util(\vecx^*) - \util(\vecx) \le \Creg \, \util(\vecx^*),
\]
where $\vecx^* = \arg\max_{\vecx \in \D} \util(\vecx)$ is the best tuple in \D.
We call the user-defined parameter \Creg the ``regret'' ratio,
a terminology used earlier in database literature \cite{nanongkai2012interactive}.
We demonstrate this setting with a concrete example below.

\begin{example}
Every tuple being a point in $\reals^3$ represents a computer with three attributes:
price, CPU speed, and hard disk capacity.
It is reasonable to assume that the utility of a computer grows linearly in, 
for example, the hard disk capacity.
Thus, a user may put a different weight on each attribute, as one entry in the utility vector $\vecu \in \reals^3$, 
which measures its relative importance.
\end{example}

For the setting described above
with a linear utility function, 
it is obvious that at most $\nD-1$ comparisons suffice to find the best tuple, 
by sequentially comparing the best tuple so far with a next tuple.
Surprisingly, despite the importance of this problem in many applications, 
improvement over the na\"ive sequential strategy, in the worst case, has remained elusive.
A positive result has only been obtained in a very restricted case of two attributes, 
i.e., a tuple is a point in $\reals^2$~\citep{xie2019strongly}.
Other existing improvements rely on strong assumptions~\citep{jamieson2011active,xie2019strongly},
for example, when every tuple is almost equally probable to be the best.
To the best of our knowledge, 
we are the first to offer an improvement on the query complexity 
that is logarithmic in \nD, in the worst case.
We refer the reader to \cref{tbl:cmp} for a detailed comparison with existing work.

There exist heuristics in the literature that are shown to perform empirically better
than the na\"ive sequential strategy, in terms of the number of requested comparisons.
For example, a popular idea is to compare a carefully-chosen pair in each round of interaction with the user~\citep{qian2015learning,wang2021interactive}.
However, these methods are computationally expensive, and require multiple passes over the whole set of tuples.
To illustrate this point, finding a ``good'' pair with respect to a given measure of interest can easily take $\bigO(\nD^2)$ time, as one has to go over ${\nD \choose 2}$ candidate pairs.
Furthermore, while such heuristics may work well in practice, 
they may require $\Omega(\nD)$ pairwise comparisons, in the worst case.

We also address the problem of finding a high-utility tuple \emph{reliably}, 
where we do not force a user to make a clear-cut decision 
when confronted with two tuples that have nearly equal utility for the user. 
In this way we can avoid error-prone decisions by a user.
Instead, we allow the user to simply declare a tie between the two tuples.
To our knowledge, this is the first paper that 
considers a scenario of finding a high-utility tuple with ties and
provides theoretical guarantees to such a scenario.

Our contributions in this paper are summarized as follows:
	($i$) We devise a \emph{single-pass} streaming algorithm 
	that processes each tuple only once, 
	and finds a high-utility tuple by making adaptive pairwise comparisons;
	($ii$) The proposed algorithm requires a memory size and has query complexity 
	that are both logarithmic in \nD, in the worst case, 
	where \nD is the number of all tuples;
	($iii$) We show how to maintain the theoretical guarantee of our method, 
	even if ties are allowed when comparing tuples with nearly equal utility;
	($iv$) We offer significant improvement to existing pruning techniques in the literature, 
	by leveraging powerful tools from mathematical programming;
	($v$) We systematically evaluate all proposed algorithms over synthetic and 
	real-life datasets, and demonstrate their superior performance compared to existing methods.

The rest of the paper is organized as follows.
We formally define the problem in \cref{se:def}.
We discuss related work in \cref{se:related}.
Then, we describe the proposed algorithm in \cref{se:single}, and 
its extension in \cref{se:robust} when ties are allowed in a comparison.
Enhancement to existing techniques follows in \cref{se:extensions-short}.
Empirical evaluation is conducted in \cref{se:experiment}, 
and we conclude in \cref{se:conclusion}.

%% file: tbl-cmp.tex
\begin{table*}
\begin{center}
\caption{\label{tbl:cmp}Comparison with existing algorithms.
We assume worst-case input with respect to the user preference and 
the distribution of the \nD tuples in the database $\D \subseteq \reals^d$,
but for the streaming case we assume that tuples arrive in random order.
An algorithm is strongly truthful if it does not use synthesized tuples that do not exist in the database \D in any comparison.
}\par\vspace{-0.3cm}
\begin{tabular}{lcccc}
\toprule
Algorithm	&\makecell{Worst-case query complexity}	&\makecell{Average-case query complexity}	&\makecell{Strongly truthful}	&Streaming	\\
\midrule
\citet{nanongkai2012interactive}	&$\bigO(d \log (d/\Creg))$	&-	&\xmark	&\xmark	\\
\citet{jamieson2011active}	&$\bigO(n)$	&$\bigO(d \log \nD)$	&\cmark	&\cmark	\\
\citet{xie2019strongly}	&$\bigO(n)$	&$\bigO(\nD^{1/d})$	&\cmark	&\xmark	\\
\cref{alg:framework,alg:single} in this paper	&$\bigO( d \log^2\!(d/\Creg) \log n)$	&-	&\cmark	&\cmark	\\
\bottomrule
\end{tabular}
\end{center}
\note{Our bound is stronger but messier $\bigO( d \log(d/\Creg) \log(d\log(d/\Creg)) \log n) = \bigO( d \log(d/\Creg) \log n \cdot [\log(d) + \log\log(d/\Creg)])$}
\end{table*}

%% file: def.tex
In this section, we formally define the \emph{interactive regret minimization} (\irm) problem.

The goal of the \irm problem is to find a good tuple among all given tuples $\D \subseteq \reals^d$ in a database.
The goodness, or utility, of a tuple \vecx is determined by an unknown utility vector $\vecu \in \reals^d$ via the dot-product operation $\util(\vecx) = \vecu^T\vecx$.
However, we assume that we do not have the means to directly compute $\util(\vecx)$, 
for a given tuple \vecx.
Instead, we assume that we have access to 
an oracle that can make comparisons between pairs of tuples:
given two tuples \vecx and \vecy
the oracle will return the tuple with the higher utility. 
These assumptions are meant to model users
who cannot quantify the utility of a given tuple on an absolute scale, 
but can perform pairwise comparisons of tuples.

In practice, it is usually acceptable to find a \emph{sufficiently good} tuple $\vecx'$ in \D, instead of the top one $\vecx^*$.
The notion of ``sufficiently good'' is measured by the ratio in utility loss 
$\frac{\util(\vecx^*) - \util(\vecx')}{\util(\vecx^*)}$, 
which is called ``regret.''
This notion leads to the definition of the \irm~problem.
\begin{problem}[Interactive Regret Minimization (\irm)]
\label{problem:IRM}
Given a set of \nD tuples in a database $\D \subseteq \reals^d$,
an unknown utility vector $\vecu \in \reals^d$, and
a parameter $\Creg \in [0,1]$,
find an \Creg-regret tuple $\vecx'$, such~that
\[
	\util(\vecx^*) - \util(\vecx') \le \Creg \, \util(\vecx^*),
\]
where $\util(\vecx) = \vecu^T\vecx$ and $\vecx^* = \arg\max_{\vecx \in \D} \util(\vecx)$.
In addition we aim at performing the minimum number of pairwise comparisons.
\end{problem}

Problem~\ref{problem:IRM} is referred to as ``interactive'' due to the fact 
that a tuple needs to be found via interactive queries to the oracle.

The parameter \Creg measures the regret.
When $\Creg = 0$, the \irm problem requires to find the top tuple $\vecx^*$ with no regret.
We refer to this special case as \emph{interactive top tuple} (\itt) problem.
For example, when tuples are in 1-dimension, 
\itt reduces to finding the maximum (or minimum) among a list of distinct~numbers.

Clearly, the definition for the \irm problem is meaningful only when $\util(\vecx^*) \ge 0$,
which is an assumption made in this paper.
Another important aspect of the \irm problem is whether or not the oracle 
will return a \emph{tie} in any pairwise comparison. 
In this paper, we study both scenarios.
In the first scenario, we assume that the oracle never returns a tie,
which implies that no two tuples in \D have the same utility.
We state our assumptions for the first (and, in this paper, default) scenario below.
We discuss how to relax this assumption for the second scenario in \cref{se:robust}.
\begin{assumption}
\label{asm:irm}
No two tuples in \D have the same utility.
Moreover, the best tuple $\vecx^*$ has non-negative utility, i.e., $\util(\vecx^*) \ge 0$.
\end{assumption}

Without loss of generality, we further assume that $\norm{\vecu}_2 = 1$ and 
$\norm{\vecx}_2 \le 1$, for all $\vecx \in \D$, which can be easily achieved by scaling.
As a consequence of our assumptions, we have $c = \util(\vecx^*) \le 1$.
The proposed method in this paper essentially finds an $\Creg/c$-regret tuple,
which is feasible for the \irm problem when $c = 1$.
Our solution still makes sense, i.e., a relatively small regret $\Creg/c$, if $c$ is not too small or a non-trivial lower bound of $c$ can be estimated in advance.
On the other hand, if $c$ is very small, 
there exists no tuple in \D that can deliver satisfactory utility in the first place,
which means that searching for the top tuple itself is also less rewarding.
For simplicity of discussion, we assume that $c = 1$ throughout the paper.

For all problems we study in this paper, 
we focus on efficient algorithms under the following computational model.
\begin{definition}[One-pass data stream model]
\label{def:streaming}
An algorithm is a one-pass streaming algorithm if 
its input is presented as a (random-order) sequence of tuples and is examined by the algorithm in a single pass.
Moreover, the algorithm has access to limited memory, generally logarithmic in the input size. 
\end{definition}
This model is particularly useful in the face of large datasets.
It is strictly more challenging than the traditional offline model, 
where one is allowed to store all tuples and examine them with random access
or over multiple passes.
A random-order data stream is a natural assumption in many applications, 
and it is required for our theoretical analysis. 
In particular, this assumption will always be met in an offline model, where one can easily simulate a random stream of tuples.
Extending our results to streams with an arbitrary order of tuples
is a major open problem.

One last remark about the \irm problem is the \emph{intrinsic dimension} of the database \D.
Tuples in \D are explicitly represented by $d$ variables, one for each dimension, and 
$d$ is called the \emph{ambient dimension}.
The intrinsic dimension of \D is the number of variables that are needed in a minimal representation of \D.
More formally,
we say that \D has an intrinsic dimension of $d'$ if 
there exist $d'$ orthonormal vectors $\myvec{b}_1, \ldots, \myvec{b}_{d'} \in \reals^d$ such that
$d'$ is minimal and
every tuple $\vecx \in \D$ can be written as a linear combination of them.
It is common that the intrinsic dimension of realistic data is much smaller than its ambient dimension.
For example, images with thousands of pixels can be compressed into a low-dimensional representation with little loss.
The proposed method in this paper is able to adapt to the intrinsic dimension of~\D 
without constructing its minimal representation explicitly.

In the rest of this section,
we review existing hardness results for the \itt and \irm problems.

\smallskip
\para{Lower bounds.}
By an information-theoretical argument, one can show that $\Omega(\log \nD)$ comparisons are necessary for the \itt problem~\citep{kulkarni1993active}.
By letting $d=\nD$ and  
$\D = \{\vece_i\}$ for $i \in [d]$, where $\vece_i$ is a vector in the standard basis,
$\Omega(d)$ comparisons are necessary to solve the \itt problem,
as a comparison between any two dimensions reveals no information about the rest dimensions.

Therefore, one can expect a general lower bound for the \irm problem to somewhat depend on both $d$ and $\log \nD$.
Thanks to the tolerance of \Creg regret in utility, 
\iftrue
a refined lower bound $\Omega(d \log (1/\Creg))$ for the \irm problem is given by \citet[Theorem 9]{nanongkai2012interactive}.
\else
in theory, one could consider an \Creg-net of the database \D, instead of the entire \D,
whose size is only $\Omega((1/\Creg)^d)$ and does not depend on \nD.
Any tuple in \D is within a distance of \Creg from some tuple in the \Creg-net.
Indeed, a refined lower bound 
$\Omega(\log(1/\Creg)^d) = \Omega(d \log (1/\Creg))$ for the \irm problem is given by \citet[Theorem 9]{nanongkai2012interactive},
by lower bounding the cardinality of the \Creg-net of a $d$-sphere.
\fi

%% file: related.tex
\para{Interactive regret minimization.}
A database system provides various operators that return a representative subset of tuples (i.e., points in $\reals^d$) to a user.
Traditional top-$k$ operators \citep{carey1997saying} return the top-$k$ tuples 
according to an explicitly specified scoring function.
In the absence of a user utility vector \vecu for a linear scoring function,
the skyline operators \citep{borzsony2001skyline} return a tuple if it has the potential to be the top tuple for at least one possible utility vector.
In the worst case, a skyline operator can return the entire dataset.
\citet{nanongkai2010regret} introduce a novel $k$-regret operator that achieves a balance between the previous two problem settings, 
by returning $k$ tuples such that the maximum regret over all possible utility vectors is~minimized.

\citet{nanongkai2012interactive} further minimize regret in an interactive fashion by making pairwise comparisons.
They prove an upper bound on the number of requested comparisons by using synthesized tuples for some comparisons.
In fact, their method learns approximately the underlying utility vector.
However, synthesized tuples are often not suitable for practical use.

\citet{jamieson2011active} deal with a more general task of finding a full ranking of \nD tuples.
By assuming that every possible ranking is equally probable, 
they show that $\bigO(d \log \nD)$ comparisons suffice to identify the full ranking in expectation.
Nevertheless, in the worst case, one cannot make such an assumption, and
their algorithm may require $\Omega(\nD^2)$ comparisons for identifying a full ranking or $\Omega(\nD)$ comparisons for identifying the top tuple.
Another similar problem assumes a distribution over the utility vector \vecu without access to the embedding of the underlying metric space~\citep{karbasi2012comparison}.
The problem of combinatorial nearest neighbor search is also related,
where one is to find the top tuple as the nearest neighbor of a given tuple \vecu without access to the embedding~\citep{haghiri2017comparison}.

\citet{xie2019strongly} observe that the \itt problem is equivalent to a special linear program, 
whose pivot step for the simplex method can be simulated by making a number of comparisons.
Thus, an immediate guarantee can be obtained by leveraging the fact that 
$\bigO(\nD^{1/d})$ pivot steps are needed in expectation for the simplex method \citep{borgwardt1982average}.
Here the expectation is taken over some distribution over \D.
Also in the special case when $d=2$, they develop an optimal binary search algorithm~\citep{xie2019strongly}. 
\citet{zheng2020sorting} suggest letting a user sort a set of displayed tuples in each round of interaction, 
but their approaches are similar to \citet{xie2019strongly}, and do not use a sorted list the way we do.

There are other attempts to the \itt problem that adaptively select a greedy pair of tuples with respect to some measure of interest.
\citet{qian2015learning} iteratively select a hyperplane (i.e., pair) whose normal vector is the most orthogonal to the current estimate of \vecu.
\citet{wang2021interactive} maintain disjoint regions of \vecu over $\reals^d$, one for each tuple, where a tuple is the best if \vecu is located within its region.
Then, they iteratively select a hyperplane that separates the remaining regions as evenly as possible.
However, these greedy strategies are highly computationally expensive, and do not have any theoretical guarantee.

Compared to aforementioned existing work,
our proposed algorithm makes minimal assumptions, is scalable, 
and enjoys the strongest worst-case guarantee.
It is worth mentioning that existing research often assumes that increasing any tuple attribute always improves utility,
by requiring $\D \subseteq \reals_+^d$ and $\vecu \in \reals_+^d$~\citep{nanongkai2012interactive,xie2019strongly,zheng2020sorting,wang2021interactive}.
We do not make such an assumption in this paper.

\smallskip
\para{Active learning.}
The \irm problem can be viewed as a special highly-imbalanced linear classification problem.
Consider a binary classification instance, where the top tuple is the only one with a positive label and the rest are all negative.
Such labeling is always realizable by a (non-homogeneous) linear hyperplane, e.g.,
$\h = \{ \vecx \in \reals^d: \vecu^T\vecx = \vecu^T\vecx^* - \Coff \}$ for any sufficiently small $\Coff \ge 0$.
Note that non-homogeneous \h can be replaced by a homogeneous one (i.e., without the offset term \Coff) by lifting the tuples into $\reals^{d+1}$.

Active learning aims to improve sample complexity that is required for learning a classifier by adaptive labeling.
Active learning with a traditional labeling oracle has been extensively studied.
The above imbalanced problem instance happens to be a difficult case for active learning with a labeling oracle \cite{dasgupta2005coarse}.
We refer the reader to \citet{hanneke2014theory} for a detailed treatment.

Active learning with additional access to pairwise comparisons has been studied by \citet{kane2017active,kane2018generalized}.
That is, one can use both labeling and comparison oracles.
Importantly, \citet{kane2017active} introduce a notion of ``inference dimension,'' with which they design an algorithm to effectively infer unknown labels.
However, due to technical conditions, the inference technique is only useful for classification in low dimension ($d \le 2$) or special instances.
As one of our main contributions, 
we are the first to show that the inference technique can be adapted for the \irm problem.

\smallskip
\para{Ranking with existing pairwise comparisons.}
A different problem setting, 
is to rank collection of tuples by aggregating a set of 
(possibly incomplete and conflicting) pairwise comparisons, 
instead of adaptively selecting which pair of tuples to compare. 
This problem has been extensively studied in the literature within different abstractions.
From a combinatorial perspective, it is known as the \emph{feedback arc-set} problem on tournaments, where the objective is to find a ranking by removing a minimum number of inconsistent comparisons~\citep{ailon2008aggregating}.
There also exist statistical approaches to find a consistent ranking, or the top-$k$ tuples, by estimating underlying preference scores~\citep{minka2018trueskill,negahban2012iterative,chen2015spectral}.
In machine learning, the problem is known as ``learning to rank'' with pairwise preferences~\citep{liu2009learning},
where the aim is to find practical ways to fit and evaluate a ranking.

\note[$X+Y$ sorting\footnote{\href{https://en.wikipedia.org/wiki/X\_\%2B\_Y\_sorting}{https://en.wikipedia.org/wiki/X\_\%2B\_Y\_sorting}}]{
The sorting approach we propose in this paper is also relevant to the $X+Y$ sorting problem~\cite{bremner2006necklaces}.
Given a sample set $S \subseteq \D$,
let $X = \{ \vecu^T \vecx : \vecx \in X\}$, $Y = -X$, and 
\[
	X+Y = \{ x+y : x \in X, y \in Y \}.
\]
Sorting $X+Y$ reveals all pairwise comparisons in $S$.
It remains open if one can do this faster than $\bigO(s^2 \log(s))$ where $s = |S|$.
}

%% file: irm.tex
In this section, we present our single-pass streaming algorithm for the \irm problem.
Our approach, presented in \cref{alg:framework,alg:single}, 
uses the concept of \emph{filters} to prune sub-optimal tuples without the need of further comparisons.
\cref{alg:framework} is a general framework for managing filters, 
while \cref{alg:single} specifies a specific filter we propose.
As we will see in \cref{se:experiment} the framework can also be used for other filters.

The filter we propose relies on a remarkable inference technique introduced by \citet{kane2017active,kane2018generalized}.
Note that the technique was originally developed for active learning in a classification task, and
its usage is restricted to low dimension ($d \le 2$) or special instances under technical conditions.
We adapt this technique to devise a provably effective filter for the \irm problem.
In addition, we strengthen their technique with a high-probability guarantee and a generalized symmetrization argument.

The core idea 
is to construct a filter from a small random sample of tuples.
It can be shown that the filter is able to identify 
a large fraction of sub-optimal tuples in \D without further comparisons.
Fixing a specific type of filter with the above property,
\cref{alg:framework} iteratively constructs a new filter in a boosting fashion 
to handle the remaining tuples.
Finally, one can show that, with high probability, 
at most $\bigO(\log \nD)$ such filters will be needed.

\input{alg-framework}

\input{alg}

We proceed to elaborate on the mechanism of a filter.
The idea is to maintain a random sample $S$ of $\sz$ tuples, and
sort them in order of their utility. 
The total order of the tuples in $S$ can be constructed by pairwise comparisons, 
e.g., by insertion sort combined with binary search.
Suppose that $S = \{\vecx_1,\ldots,\vecx_\sz\}$, where $\vecx_1$ has the best utility.
Notice that $\vecu^T(\vecx_{j+1} - \vecx_{j}) \le 0$ for any $j$.
Thus, a sufficient condition for an arbitrary tuple \vecx to be worse than $\vecx_1$ is
\begin{align}
\label{eq:filter:lp}
	\vecx = \vecx_1 + \sum_{j=1}^{\sz-1} \alpha_j (\vecx_{j+1} - \vecx_j)
	\quad
	\text{such that} \; \alpha_j \ge 0 \quad \text{for all } j.
\end{align}

This condition asks to verify whether \vecx lies within a cone with apex $\vecx_1$, 
along direction \vecu.
The parameters $\alpha_j$ can be efficiently computed by a standard Linear Program (LP) solver.
If Condition~(\ref{eq:filter:lp}) can be satisfied for \vecx,
then \vecx can be pruned for further consideration.

Actually, it is possible to act more aggressively and prune tuples slightly better than $\vecx_1$, 
as long as it is assured that not all feasible tuples will be pruned.
Specifically, we can remove any  \vecx that deviates from the aforementioned cone within a distance of $\Creg$, that~is,
\begin{align}
	\label{eq:filter}
	&\min_\alpha \; \bigg\| \vecx - \vecx_1 - \sum_{j=1}^{\sz-1} \alpha_j (\vecx_{j+1} - \vecx_j) \bigg\|_2 \le \Creg 
	\ \ \  \text{s.t.} \ \  \alpha_j \ge 0 \  \text{for all } j .
\end{align}

To test whether a given tuple \vecx satisfies the above condition, 
one needs to search for parameters $\alpha_j$ over $[0,\infty)$ for all $j$.
The search can be implemented as an instance of constrained least squares, 
which can be efficiently solved via a quadratic program~(QP).

Given a sorted sample $S$ where $\vecx_1$ is the top tuple, we write 
\begin{align}
	\label{eq:implies}
	S \infer \vecx
\end{align}
if a tuple \vecx can be approximately represented by vectors in $S$ in a form of \cref{eq:filter}.

\begin{figure}
	\centering
	\includegraphics[width=0.3\textwidth]{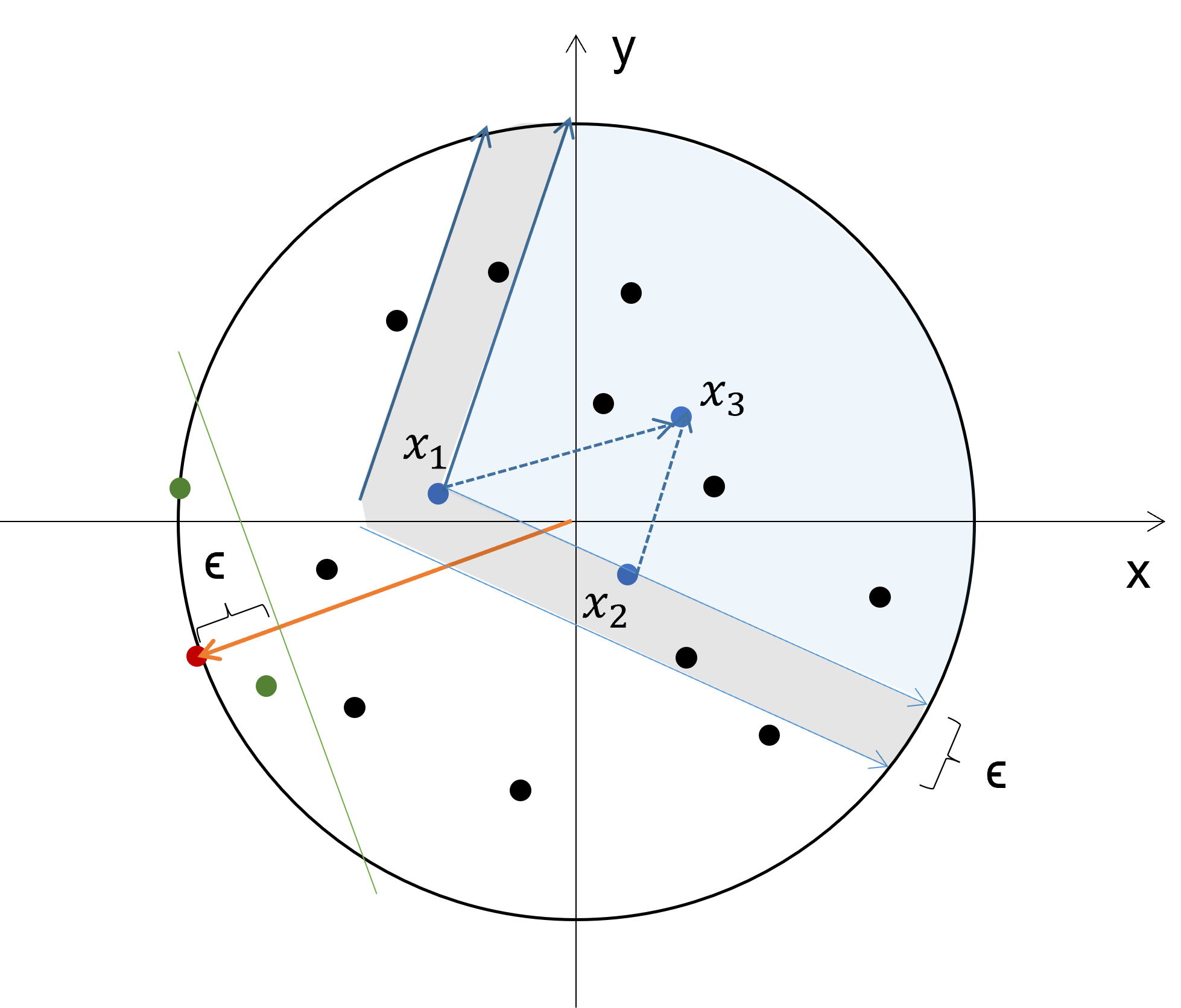}\par\vspace{-0.2cm}
	\caption{An illustrative example for a filter in \cref{alg:single}.
	The unknown utility vector \vecu is drawn in orange.
	Every tuple is shown as a point within a unit circle, 
	where the red point represents the top tuple, and 
	green points are feasible \Creg-regret tuples for the \irm problem.
	Suppose a filter collects a random sample $S$ of blue points.
	A sorted sample $S$ can be used to prune any point that falls into or is sufficiently close to (within a distance of \Creg) the blue cone.
	}\label{fig:filter}
\end{figure}

An example that illustrates the mechanism of a filter is displayed in \cref{fig:filter},
on which we elaborate below.
\begin{example}
In \cref{fig:filter},
a random sample $S = \{\vecx_1,\vecx_2,\vecx_3\}$ of three blue points is collected and sorted,
where $\vecx_1$ has the highest utility.
This means that 
$\util(\vecx_{j+1}) - \util(\vecx_{j}) = \vecu^T (\vecx_{j+1} - \vecx_j) < 0$, for any $j \in \{1,2\}$.
Compared to the point $\vecx_1$, 
a new point \vecx in the form of 
$\vecx = \vecx_1 + \sum_{j \in \{1,2\}} \alpha_j (\vecx_{j+1} - \vecx_j)$ with $\alpha_j \ge 0$ 
can only have a lower utility than $\util(\vecx_1)$,
since
\[
	\util(\vecx) = \vecu^T \bigg[ \vecx_1 + \sum_{j \in \{1,2\}} \alpha_j (\vecx_{j+1} - \vecx_j) \bigg] \le \util(\vecx_1).
\]
Thus, such a point \vecx can be safely pruned.
Geometrically, all such prunable points \vecx form a cone with apex $\vecx_1$,
as highlighted in the blue region in \cref{fig:filter}.
According to $\cref{eq:filter}$, 
any point that is sufficiently close to (within a distance of \Creg) the blue cone can also be pruned.
\end{example}

Upon a random-order stream of tuples, \cref{alg:framework,alg:single} collect a pool \Pl of \nPl initial tuples as a testbed for filter performance.
Then, subsequent tuples are gradually added into the first sample set $S_1$,
until a filter based on $S_1$ can prune at least a $\CPl = 5/8$ fraction of \Pl.
Then, $S_1$ is ready, and is used to prune tuples in the pool \Pl and future tuples over the stream.
Future tuples that survive the filter formed by $S_1$ will be gradually added into the pool \Pl and a second sample set $S_2$, and the process is repeated iteratively.
Finally, the algorithm returns the best tuple among all samples.
The following theorem states our main result about \cref{alg:framework,alg:single}.

\begin{theorem}
\label{thm:num_of_S}
Assume $\epsilon > 0$ and let $\nD = |\D|$ be the size of data.
Let $c = \util(\vecx^*) \in [0,1]$ be the utility of the best tuple~$\vecx^*$.
Under \cref{asm:irm},
with a pool size $\nPl = \ceil{64 \ln 2\nD}$ and $\CPl = 5/8$,
\cref{alg:framework,alg:single} return an $\Creg/c$-regret tuple for the \irm problem.

Let $\infd = \Cinfd$, where
$d$ is the intrinsic dimension of~\D.
Then, with probability at least $1-1/\nD$, at most
\[
	\bigO(\log(\nD) \, 4t \log (4t)) + \nPl
\]
comparisons
are made.
\end{theorem}
The memory size, i.e., the number of tuples that will be kept by the algorithm during the execution, 
is $\bigO(\log(\nD) \, 4t)$, 
which is also logarithmic in~\nD.

In fact, \cref{alg:framework,alg:single} are an \emph{anytime} algorithm,
in the sense that the data stream can be stopped anytime,
while the algorithm is still able to return a feasible solution among all tuples that have arrived so far.

\begin{theorem}
\label{thm:anytime}
Under \cref{asm:irm},
the data stream may terminate at any moment during the execution of \cref{alg:framework,alg:single}, and
an $\Creg/c$-regret tuple will be returned for the \irm problem among all tuples that have arrived so far.
\end{theorem}

Proofs of \cref{thm:num_of_S,thm:anytime} are deferred to \cref{se:single:proofs}.

%% file: alg-framework.tex
\begin{algorithm}[t]
\DontPrintSemicolon
\SetKwFunction{FNewFilter}{NewFilter}
\KwIn{tuples \D and parameters \nPl, \CPl}
$F \gets \FNewFilter, \, \seq \gets (), \, \Pl \gets \emptyset$ \;
\For{tuple $\vecx \in \D$ over a random-order stream}{
	\lIf{$F'.prune(\vecx)$ for any $F' \in \seq$}{
		\Continue 
	}
	\lIf{$|\Pl| < \nPl$}{
		$\Pl \gets \Pl \cup \{\vecx\} $;
		\Continue
	}
	$F.add(\vecx)$ \;
	$\Pl' \gets \{ \vecy \in \Pl: F.prune(\vecy) \text{ is true} \}$ \;
	\If{$|\Pl'| \ge \CPl |\Pl|$}{ \label{step:pool}
		Append $F$ to sequence \seq \;
		$F \gets \FNewFilter, \, \Pl \gets \Pl \setminus \Pl'$ \;
	}
}
Append $F$ to sequence \seq and let $\X = \{ F.best() : F \in \seq \}$ \;
Let \bestsofar be the best tuple in $\X \cup \Pl$ by pairwise comparisons \;
\Return \bestsofar \;
\caption{A general framework}
\label{alg:framework}
\end{algorithm}

%% file: alg.tex
\begin{algorithm}[t]
\DontPrintSemicolon
\KwIn{parameter \Creg}
\lClass{\FNewFilter}{
	$S \gets \emptyset$
}

\Function{prune$(\vecx)$}{
	\Return true, if $S \infer \vecx$ (see \cref{eq:implies}), otherwise false \label{step:filter}\;
}

\Function{add$(\vecx)$}{
	$S \gets S \cup \{ \vecx \}$ and sort $S$ by pairwise comparisons \;
}

\lFunction{best$()$}{
	\Return the best tuple $\vecx_1$ in $S$ 
}

\caption{Functions that define a filter for the \irm problem with no ties}
\label{alg:single}
\end{algorithm}

%% file: robust.tex
In this section, 
we first introduce a natural notion of uncomparable pairs to avoid error-prone comparisons, and then
we show how this new setting affects our algorithms.

It is clearly more difficult for a user to distinguish a pair of tuples with nearly equal utility.
Thus, it is reasonable to not force the user to make a choice in the face of a close pair, and 
allow the user to simply declare the comparison a tie instead.
We make this intuition formal below.

\begin{definition}[\vecu-similar pairs]
\label{def:cmp}
Two tuples $\vecx,\vecy \in D$ are \vecu-similar if 
\[
	|\util(\vecx) - \util(\vecy)| \le \Ccmp,
\]
for some fixed value \Ccmp.
We write $\vecx \cmp \vecy$ if they are uncomparable.
\end{definition}

\begin{assumption}
\label{asm:dirm}
A query about a \vecu-similar pair to the oracle will be answered with a tie.
Besides, as before, we assume that the best tuple~$\vecx^*$ has non-negative utility, 
$\util(\vecx^*) \ge 0$.
\end{assumption}

Typically, the value \Ccmp is fixed by nature, unknown to us, and cannot be controlled.
Note that when \Ccmp is sufficiently small, we recover the previous case in \cref{se:single} where every pair is comparable under \cref{asm:irm}.
By allowing the user to not make a clear-cut comparison for a \vecu-similar pair, one can no longer be guaranteed total sorting.
Indeed, it could be that every pair in \D is \vecu-similar.

In \cref{alg:robust}, we provide a filter to handle ties under \cref{asm:dirm}.
We maintain a totally sorted subset $R$ of \emph{representative} tuples in a sample set $S$.
For each representative $\vecy \in R$, we create a group~$G_\vecy$.
Upon the arrival of a new tuple \vecx, we sort \vecx into $R$ if no tie is encountered.
Otherwise, we encounter a tie with a tuple $\vecy \in R$ such that $\vecx \cmp \vecy$, and
we add \vecx into a group $G_\vecy$.
In the end, the best tuple in $R$ will be returned.

\input{alg-robust}

To see whether a filter in \cref{alg:robust} can prune a given tuple~\vecx, 
we test the following condition.
Let $R = \vecx_1, \ldots$ be the sorted list of representive tuples, 
where $\vecx_1$ is the top tuple.
Let $\Gs = G_1, \ldots$ be the corresponding groups.
A tuple \vecx can be pruned if 
there exists $\vecx'$ such that $\norm{\vecx - \vecx'}_2 \le \Creg$, where
\begin{align}
	\vecx' =& 
	\sum_{\vecy \in G_{1} \cup G_{2}} \nu_\vecy \, \vecy  
	+ \sum_{j=1} \sum_{ \vecz \in G_{j} } \sum_{\vecw \in  G_{{j+2}}} \alpha_{\vecw,\vecz} (\vecw - \vecz) \label{eq:filter-robust}\\
	&\text{such that} 
	\sum_{\vecy \in G_{1} \cup G_{2}} \nu_\vecy = 1 \text{ and all } \nu, \alpha \ge 0 \nonumber .
\end{align}
The idea is similar to \cref{eq:filter}, except that 
the top tuple $\vecx_1$ in \cref{eq:filter} is replaced by an aggregated tuple by convex combination, and
every pair difference $\vecx_{j+1} - \vecx_{j}$ is replaced by pair differences between two groups.
We avoid using pair differences between two consecutive groups, as tuples in group $G_{j}$ may not have higher utility than tuples in $G_{{j+1}}$.
If the above condition is met, then
we write 
\begin{align}
	\Gs \infersim \vecx
	\quad\text{and, if $\Gs$ is constructed using $S$,}\quad
	S \infersim \vecx \label{eq:implies-robust}.
\end{align}

The number of comparisons that is needed by \cref{alg:robust} depends on the actual input,
specifically, \Ccmpmax, the largest size of any pairwise \vecu-similar subset of~\D.
Note that the guarantee below recovers that of \cref{thm:num_of_S} up to a constant factor, if assuming \cref{asm:irm} where $\Ccmpmax=1$.
However, in the worst case, $\Ccmpmax = \bigO(\nD)$ and the guarantee becomes vacuous.
\begin{theorem}
\label{thm:num_of_S-robust}
Assume $\epsilon > 0$ and let $\nD = |\D|$ be the size of data.
Let $c = \util(\vecx^*) \in [0,1]$ be the utility of the best tuple~$\vecx^*$.
Under \cref{asm:dirm},
with a pool size $\nPl = \ceil{256\ln 2\nD}$ and $\CPl = 3/16$,
\cref{alg:framework,alg:robust} return an $(\Creg/c + 2\Ccmp)$-regret tuple for the \irm problem.

Let $\infd = \Cinfd$, where
$d$ is the intrinsic dimension of~\D, and
\Ccmpmax be the largest size of a pairwise \vecu-similar subset of \D.
Then, with probability at least $1-1/\nD$, at most
\[
	\bigO(\log(\nD) \, 16 \infd \Ccmpmax \log(16 \infd \Ccmpmax)) + \nPl
\]
comparisons
are made.
\end{theorem}

Proofs of \cref{thm:num_of_S-robust} are deferred to \cref{se:robust:proofs}.

%% file: alg-robust.tex
\begin{algorithm}[t]
\DontPrintSemicolon
\KwIn{parameter \Creg}
\lClass{\FNewFilter}{
	$R \gets \emptyset, \, S \gets \emptyset$;
}

\Function{prune$(\vecx)$}{
	\Return true, if $S \infersim \vecx$ (see \cref{eq:implies-robust}), otherwise false\;
}

\Function{add$(\vecx)$}{
	$S \gets S \cup \{\vecx\}$, and sort \vecx within $R$ \; 
	\eIf{no tie happens}{
		$R \gets R \cup \{\vecx\}$ and create $G_\vecx \gets \{\vecx\}$ \;
	}{
		Encounter a tie with $\vecy \in R$ such that $\vecx \cmp \vecy$ \;
		$G_\vecy \gets G_\vecy \cup \{\vecx\}$\;
	}
}

\lFunction{best$()$}{
	\Return the best tuple in $R$ 
}

\caption{Functions that define a filter for the \irm problem with ties}
\label{alg:robust}
\end{algorithm}

%% file: extensions-short.tex
In this section, 
we improve existing filters by \citet{xie2019strongly},
by using linear and quadratic programs.
We will use these baselines in the experiments.
Previously, their filters rely on explicit computation of convex hulls,
which is feasible only in very low dimension~\citep{barber1996quickhull}.
Technical details are deferred to \cref{se:extensions}.

Existing filters iteratively compare a pair of random tuples, all of which are kept in 
$A = \{a_i\}$,
where $a_i = (\vecy,\vecz)$ such that $\util(\vecy) < \util(\vecz)$, 
and use them to prune potential tuples.

\paragraph{Filter by constrained utility space}

Given a tuple \vecx, we try to find 
a vector $\vecu$ that, for all $ (\vecy,\vecz) \in A$,
\begin{equation}
\restatableeq{\eqlp}{
	\vecu^T(\vecz-\vecy) \ge 1,\quad
	\vecu^T(\vecx-\vecz) \ge 1,\quad
	\vecu^T((1-\Creg)\vecx - \vecz) \ge 1.
}{eq:lp}
\end{equation}

We claim that a given tuple \vecx can be safely pruned if there is no vector~$\vecu$ satisfying LP~(\cref{eq:lp}).

\begin{restatable}{proposition}{proplp}
	\label{prop:lp}
	Consider a tuple $\vecx$ with $\util(\vecx) > \util(\vecz)$ and
	$\util(\vecx) - \util(\vecz) > \Creg \util(\vecx)$ for every $(\vecy, \vecz) \in A$.
	Then there is a solution to LP~(\cref{eq:lp}).
\end{restatable}

\paragraph{Filter by conical hull of pairs}

Given a tuple \vecx, we propose to solve the following quadratic program (QP),
\begin{align}
\restatableeq{\eqqppairs}{
	& \min_{\nu,\beta} \bigg\| \vecx - \sum_{a_i = (\vecy,\vecz) \in A} (\nu_{i1} \, \vecy + \nu_{i2} \, \vecz) - \sum_{a_i = (\vecy,\vecz) \in A} \beta_i (\vecy-\vecz) \bigg\| \\
	& \text{such that}
	\sum_{a_i = (\vecy,\vecz) \in A} \nu_{i1}+\nu_{i2} = 1 
	 \quad\text{and}\quad
	\nu_{i1}, \nu_{i2}, \beta_i \ge 0 \quad\text{for all } i \nonumber
}{eq:qp-pairs}.
\end{align}

If the optimal value of the QP is at most $\Creg$, we prune \vecx.

\begin{restatable}{proposition}{propqppairs}
	\label{prop:qp-pairs}
	Let $\vecu^T \vecx^* = c$.
	A tuple $\vecx \in \D$ can be pruned if 
	the objective value of the quadratic program~(\cref{eq:qp-pairs}) is at most $\Creg / c$.
\end{restatable}

If we set $\epsilon = 0$, then we can use LP solver (similar to \cref{eq:filter:lp}) instead of QP solver.
This results in a weaker but computationally more efficient filter.

%% file: experiment.tex
In this section, we evaluate key aspects of our method and the proposed filters.
Less important experiments and additional details are deferred to \cref{se:experiment:supp}.
In particular, we investigate the following questions.
	($i$) How accurate is the theoretical bound in \cref{lemma:filter}? 
    More specifically, we want to quantify the sample size required by \cref{alg:single} to prune at least half of the tuples, 
    and understand its dependance on the data size \nD, dimension $d$, 
    and regret parameter \Creg. 
    ($ii$) Effect of parameters of \cref{alg:framework}.
	 (\cref{sse:ex:param})
	($iii$) How scalable are the proposed filters?
	($iv$) How do the proposed filters perform over real-life datasets?
	($v$) How do ties in comparisons affect the performance of the proposed filters?
Our implementation is available at \code

\smallskip
\noindent
Next, let us introduce the adopted datasets and baselines.

\smallskip
\para{Datasets.}
A summary of the real-life datasets we use for our evaluation can be found in \cref{tbl:datasets}.
To have more flexible control over the data parameters,
we additionally generate the following two types of synthesized data.
	\emph{sphere}: Points sampled from the unit $d$-sphere \sphere uniformly at random.
	\emph{clusters}: Normally distributed clustered data, where each cluster is centered at a random point on unit $d$-sphere \sphere.
To simulate an oracle, we generate a random utility vector \vecu on the unit $d$-sphere for every run.
More details about datasets can be found in \cref{se:experiment:supp}.

\smallskip
\para{Baselines.}
A summary of all algorithms is given in \cref{tbl:algs}.
We mainly compare with (enhanced) pruning techniques (\pairqp, \pairlp and \hplp) by \citet{xie2019strongly}, 
halfspace-based pruning (\hp), and 
a random baseline (\rand).
Discussion of other baselines is deferred to \cref{se:experiment:supp}.
We instantiate every filter (except for the \hp and \rand) in the framework provided in \cref{alg:framework},
that is, we iteratively create a new filter that can prune about half of the remaining tuples.
This is a reasonable strategy, and will be justified in detail in \cref{sse:ex:scalable}.
For pair-based filters, a new pair is made after two consecutive calls of the \emph{add} function.
The pool size \nPl and threshold \CPl in \cref{alg:framework} are set to be 100 and 0.5, respectively.
Since the proposed algorithm \listqp only guarantees a regret of $\Creg/\util(x^*)$, where $x^*$ is the best tuple in the dataset,
we pre-compute the value of $\util(x^*) \in [0,1]$, and adjust the regret parameter of \listqp to be $\Creg \util(x^*)$.

\begin{table}[t]
  \caption{Real-life datasets statistics}
  \label{tbl:datasets}
  \centering
  \vspace{-0.3cm}
\begin{tabular}{lrr}
\toprule
Dataset & $\nD = |\D|$ 
        & $d$ \\
\midrule
player \citep{dataplayer}	&17 386	&20	\\
youtube \citep{datayoutube}	&29 406	&50	\\
game \citep{datagame}	&60 496	&100	\\
house \citep{datahouse}	&303 032	&78	\\
car \citep{datacar}	&1 002 350	&21	\\
\bottomrule
\end{tabular}
\end{table}

\begin{table}[t]
  \caption{Summary of our methods and the baselines.}
  \label{tbl:algs}
  \centering
  \vspace{-0.3cm}
\begin{tabular}{l p{0.35\textwidth}}
\toprule
Name & Brief description\\
\midrule
List-[QP|LP]	&\emph{Our method}: prune a tuple if it is close to a conical hull formed by a sorted list of random tuples (\cref{alg:single}), equipped with a QP or LP solver.	\smallskip\\
Pair-[QP|LP]	&Prune a tuple if it is close to a conical hull formed by a set of compared random pairs, equipped with QP (\cref{eq:qp-pairs}) or LP 
solver. \smallskip\\
\hplp	&Prune a tuple if LP (\cref{eq:lp}) is infeasible, i.e.,
the tuple is dominated by a set of compared random pairs over the entire constrained utility space.	\smallskip\\
\hp	&Prune a tuple \vecx if $\vecx^T(\vecz-\vecy) < 0$ for any compared pair $\vecz,\vecy$ such that $\util(\vecy) < \util(\vecz)$, that is,
tuple \vecx falls outside the constrained utility space for \vecu. \smallskip\\
\rand	&Return the best tuple among a subset of 50 random tuples.	\\
\bottomrule
\end{tabular}
\end{table}

\begin{figure}[!ht]
    \centering
    \subcaptionbox{Memory size required to prune half of tuples\label{fig:n-sz}}{
    	\pgfinput[trim={0 0 0 0.1cm},clip]{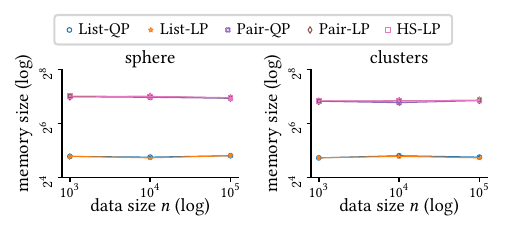}\vspace{-0.3cm}
    }\par\vspace{0.2cm}
    \subcaptionbox{Scalability with respect to dimension $d$\label{fig:d-runtime}}{
  	    \pgfinput[trim={0 0 0 0.8cm},clip]{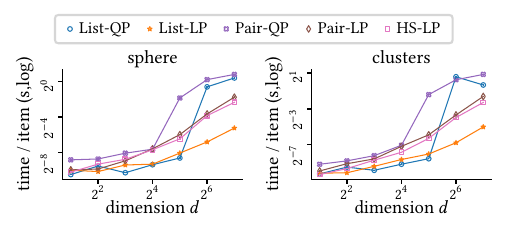}\vspace{-0.3cm}
    }
    \vspace{-0.3cm}
    \caption{Scalability of filters for synthetic data}
    \label{fig:ex}
\end{figure}

\begin{figure*}[t]
    \centering
   	\pgfinput[trim={0 0 0 0.1cm},clip]{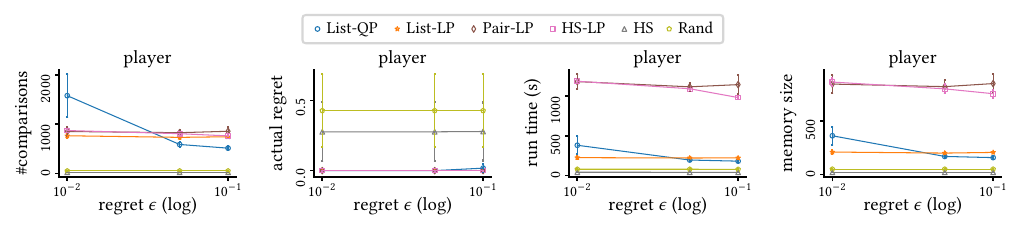}\par\vspace{-0.3cm}
   	\pgfinput[trim={0 0 0 0.8cm},clip]{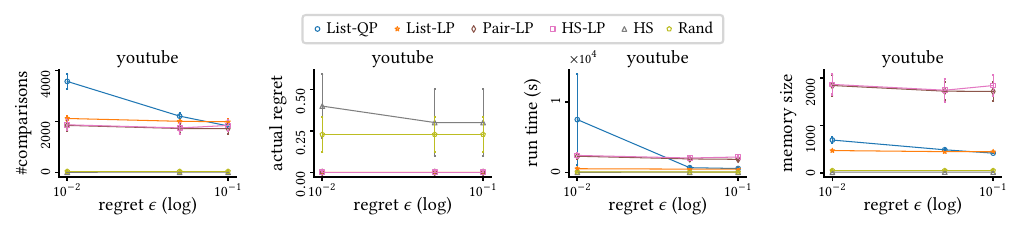}\par\vspace{-0.3cm}
   	\pgfinput[trim={0 0 0 0.8cm},clip]{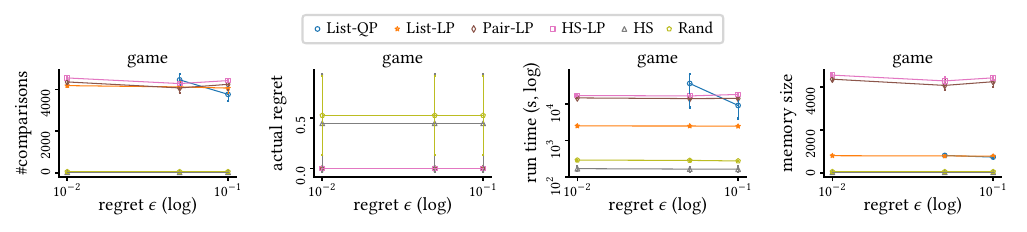}\par\vspace{-0.3cm}
 	\pgfinput[trim={0 0 0 0.8cm},clip]{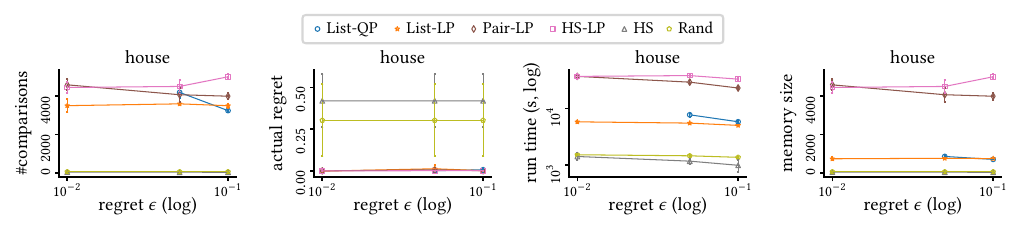}\par\vspace{-0.3cm}
 	\pgfinput[trim={0 0 0 0.8cm},clip]{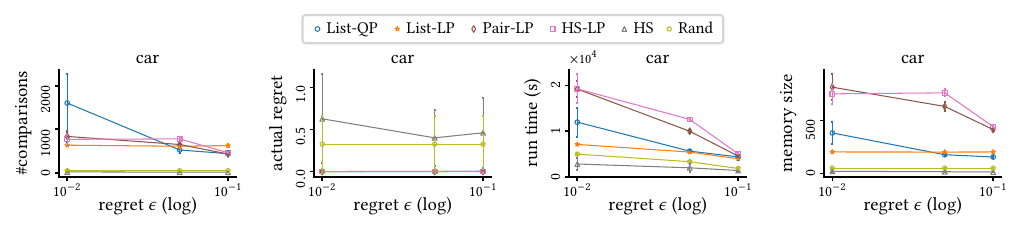}
 	\vspace{-0.6cm}
    \caption{Solving the \irm problem on the real-life datasets}
    \label{fig:ex:single}
\end{figure*}

\begin{figure*}[t]
    \centering
    \vspace{-0.4cm}
    \pgfinput{gap-gap-ncmp}
    \vspace{-0.6cm}
    \caption{The effect of ties and parameter \Ccmp}
    \label{fig:ties}
\end{figure*}

\subsection{Sample size in practice}
\label{sse:ex:sameplesize}

\cref{lemma:filter} proves a theoretical bound on the size of a random sample 
required by \cref{alg:single} to prune at least half of a given set \D of tuples in expectation.
This bound is $2\infd$ where $\infd = \Cinfd$.
Importantly, the bound does not depend on the data size $|\D|$,
which we verify later in~\cref{sse:ex:scalable}.

In \cref{fig:d-sz,fig:tol-sz} (in Appendix), we compute and present the exact required size for synthesized data, and illustrate
how the size changes with respect to the dimension $d$ and regret parameter \Creg.
As can be seen, the bound provided in \cref{lemma:filter} captures a reasonably accurate dependence on $d$ and \Creg, up to a constant factor.

\subsection{Scalability}
\label{sse:ex:scalable}

The running time required for each filter to prune a given tuple depends heavily on its memory size,
i.e., the number of tuples it keeps.
In \cref{fig:n-sz}, we compute and show the required memory size for a filter to prune half of a given set \D of tuples, and
how the size changes with respect to the data size $\nD = |\D|$.
Impressively, most competing filters that adopt a randomized approach only require constant memory size, regardless of the data size \nD.
This also confirms the effectiveness of randomized algorithms in pruning.

Based on the above observation,
it is usually not feasible to maintain a single filter to process a large dataset \D.
If a filter requires~$s$ tuples in memory to prune half of \D,
then at least $s \log(|\D|)$ tuples are expected to process the whole dataset \D.
However, the running time for both LP and QP solvers is superlinear in the memory size of a filter~\citep{cohen2021solving}, 
which means that running a filter with $s \log(|\D|)$ tuples is considerably slower than running $\log(|\D|)$ filters, each with $s$ tuples.
The latter approach enables also parallel computing for faster processing.

Therefore, we instantiate each competing filter (except for \hp and \rand) in the framework provided in \cref{alg:framework}, and
measure the running time it takes to solve the \irm problem.
In the rest of this section, 
we investigate the effect of the data dimension $d$ and regret parameter \Creg on the running time.

\smallskip
\para{Effect of data dimension $d$.}
In \cref{fig:d-runtime}, 
we fix a regret parameter $\Creg=0.01$, and
examine how the running time of a filter varies with respect to the data dimension $d$ on synthesized data.

The first observation from \cref{fig:d-runtime} is that LP-based filters are more efficient than their QP counterparts.
Particularly, \pairqp is too slow to be used, and we have to settle for its LP counterpart \pairlp in subsequent experiments.

Let us limit the comparison to those LP-based filters.
\pairlp and \hplp are more computationally expensive than \listlp.
For \pairlp, the reason is obvious: 
as discussed at the end of \cref{sse:cone-by-pairs}, 
\pairlp makes relatively more comparisons and 
every compared pair of tuples adds two more parameters to the LP.
For \hplp, the number of parameters in its LP depends linearly on both the dimension $d$ and number of compared pairs,
while \listlp only depends on the latter.
Thus, \hplp is less scalable by design.
\note{When $d$ increases, the impact of \#parameters needed for compared pairs dominates that of those modeling \vecu for \hplp, so its running time converges to \listlp.}

\smallskip
\para{Effect of regret parameter \Creg.}
The effect of the regret parameter \Creg can be found in \cref{fig:ex:single} for all real-life datasets.
Generally, a larger value of \Creg decreases the running time, 
as each filter can be benefited by more aggressive pruning.

The running time of \listqp deteriorates dramatically for a small value of \Creg, and
the number of comparisons needed also rises considerably.
The reason is that,
most numerical methods for solving a mathematical program have a user-defined \emph{precision} parameter.
Small precision gives a more accurate solution, and at the same time causes a longer running time.
When \Creg gets close to the default precision, or to the actual precision after the maximum number of iterations is exceeded, \listqp fails to prune tuples.
Thus, \listqp is advised to be used for a relatively large regret value \Creg.

In regard to the memory size,
as we can see in \cref{fig:ex:single},
\listqp and \listlp consistently use a much smaller memory size than \pairlp and \hplp.
This also demonstrates the advantage of using a sorted list over a set of compared pairs.

\subsection{The case of oracles with no ties}
\label{sse:ex:single}

The performance of competing filters can be found in \cref{fig:ex:single} for all real-life datasets.
The average and standard error of three random runs are reported.
We instantiate each competing filter (except for \hp and \rand) in the framework provided in \cref{alg:framework} to solve the \irm problem.
Meanwhile, we vary the regret parameter \Creg to analyze its effect.
We also experimented with a smaller \Creg value such as 0.005, 
the observations are similar except that the \listqp filter is significantly slower for reasons we mentioned in \cref{sse:ex:scalable}.

Except \hp and \rand,
every reasonable filter succeeds in returning a low-regret tuple.
We limit our discussion to only these reasonable filters.
In terms of the number of comparisons needed,
\listqp outperforms the rest on most datasets provided that the regret value \Creg is not too small.
We rate \listlp as the runner-up, and it becomes the top one when the regret value \Creg is small.
Besides, \listlp is the fastest to run.
The number of comparisons needed by \hplp and \pairlp is similar, and 
they sometimes perform better than others, for example, over the youtube dataset.

Let us make a remark about the regret value \Creg.
Being able to exploit a large value of \Creg in pruning is the key to improving performance.
Notice that both \pairlp and \listlp cannot benefit from a large regret value \Creg by design.
Though \hplp is designed with \Creg in mind, 
it is more conservative as its pruning power depends on $\Creg \vecu^T\vecx$ instead of $\Creg \vecu^T\vecx^*$,
where \vecx is the tuple to prune.

In summary, we can conclude that
the \listqp filter is recommended for a not too small regret parameter \Creg (i.e., $\Creg \ge 0.1$), and
the \listlp filter is recommended otherwise.
In practice, since both \listqp and \listlp follow an almost identical procedure, 
one could always start with \listqp, and switch to \listlp if the pruning takes too long time. 

\todo{I cannot explain why in house dataset, \#comparison goes up for \listlp and \hplp by increasing \Creg. 
This also happens mildly in game dataset. 
One common feature of them is having a lot of one-hot attributes.}

\subsection{Effect of ties}
\label{sse:ex:ties}

According to \cref{asm:dirm},
the oracle returns a tie if the difference in utility between two given tuples is within a parameter \Ccmp.
For filters like \pairlp and \hplp, the most natural strategy to handle a tie for a pair of tuples is to simply discard one of them.
It is expected that ties worsen the performance of a filter, as they fail to provide
additional information required by the method for pruning.

In \cref{fig:ties}, we vary the value of parameter \Ccmp to see 
how it affects the performance of the proposed filters.
It is not surprising that
as the value of \Ccmp increases, the number of ties encountered and the number of comparisons made by all algorithms both increase.

Notably, the running time of \listqp and \listlp grows significantly as \Ccmp increases. 
This is because one parameter is needed in their solvers for every pair of tuples between two consecutive groups $G_i,G_j$, and
the total number of parameters can increase significantly if the size of both groups increases.
This behavior also reflects the fact that a partially sorted list is less effective for pruning.
However, how to handle a large \Ccmp remains a major open problem.
Hence, we conclude that the proposed algorithms work well provided that the parameter \Ccmp is not too large.

\paragraph{Summary}
\label{sse:ex:summary}

After the systematical evaluation,
we conclude with the following results.
	($i$) LP-based filters are more efficient than their QP counterparts, but less effective in pruning.
	($ii$) \listlp is the most scalable filter.
	The runner-up is \listqp, provided that the data dimension is not too large ($d < 128$) and the regret parameter \Creg is not too small ($\Creg \ge 0.1$).
	($iii$) To minimize the number of requested comparisons,
	\listqp is recommended for a not too small \Creg ($\Creg \ge 0.1$).
	When \Creg is small, we recommend \listlp.
	($iv$) Good performance can be retained if the oracle is sufficiently discerning ($\Ccmp \le 0.01$).
	Otherwise, a better way to handle ties will be needed.

\todo[Minor points]{
\begin{itemize}
	\item Handle (few) rounding error in comparison
	\item Run for offline settings
	\item Combine different kinds of filters
\end{itemize}
}

%% file: conclusion.tex
We devise a single-pass streaming algorithm for finding a high-utility tuple by making adaptive pairwise comparisons.
We also show how to maintain the guarantee when ties are allowed in a comparison between two tuples with nearly equal utility.
%
Our work suggests several future directions to be explored.
Those include
finding a high-utility tuple in the presence of noise,
incorporating more general functions for modeling tuple utility, 
devising methods with provable quarantees for arbitrary-order data streams,
and
devising more efficient algorithms to handle ties.